%% file: Sabino_OU_VG.tex
\documentclass[showpacs,preprintnumbers,amsmath,amssymb,12pt]{article}

\usepackage{a4wide,amsmath,amsthm,epsfig,graphicx}
\usepackage{amsmath,amsthm,amssymb,eurosym}
\usepackage{amsfonts}
\usepackage{graphics}
\usepackage{adjustbox}
\usepackage{hyperref}
\usepackage{dsfont}
\usepackage{bbm}
\usepackage{bm}
\usepackage{xcolor}
\usepackage{dcolumn}
\usepackage{fancyhdr,fancybox}
\usepackage[svgnames]{pstricks}
\usepackage{pst-text}
\usepackage{pst-plot,pst-eucl,pstricks-add}
\usepackage{caption}
\usepackage{subcaption}

\usepackage{algorithm}
\usepackage{algpseudocode}
\usepackage{cool}
\usepackage{booktabs}
\usepackage{makecell}

\newtheorem{thm}{Theorem}[section]

\newtheorem{prop}[thm]{Proposition}

\newcommand{\sde}{\emph{SDE}}

\newcommand{\refeq}[1]{~(\ref{#1})}
\newcommand{\myref}[1]{~\ref{#1}}
\newcommand{\mycite}[1]{~\cite{#1}}

\newcommand{\R}{\mathbb{R}}

\newcommand{\rv}{\textit{rv}}
\newcommand{\id}{\textit{id}}
\newcommand{\iid}{\textit{iid}}
\newcommand{\sd}{\textit{sd}}

\newcommand{\pdf}{\textit{pdf}}

\newcommand{\chf}{\textit{chf}}
\newcommand{\che}{\textit{che}}

\newcommand{\Levy}{L\'{e}vy}

\newcommand{\Pqo}{\bm{P}\hbox{-\emph{a.s.}}}
\newcommand{\eqd}{\stackrel{d}{=}}

\newcommand{\EXP}[1]{\bm{E}\left[{#1}\right]}
\newcommand{\VAR}[1]{\bm{V}\left[{#1}\right]}
\newcommand{\SK}[1]{\bm{Skew}\left[{#1}\right]}
\newcommand{\KUR}[1]{\bm{Kurt}\left[{#1}\right]}

\newcommand{\lap}{\mathfrak{La}}

\newcommand{\poiss}{\mathfrak{P}}

\newcommand{\unif}{\mathfrak{U}}

\newcommand{\erl}{\mathfrak{E}}

\newcommand{\ou}{\emph{OU}}

\newcommand{\arem}{$a$-remainder}
\newcommand{\dilog}{\text{Li}_2}

\DeclareMathOperator{\sign}{sign}

\usepackage[a4 paper, top = 1.3 cm, bottom = 1.5 cm, left = 1.3 cm, right = 1.3 cm]{geometry}
\title{\Huge \textbf{Exact Simulation of Variance Gamma related OU processes: \\
Application to the Pricing of Energy Derivatives.}\footnote{The views, opinions, positions or strategies expressed in this work are those of the author and do not represent the views, opinions and strategies of, and should not be attributed to E.ON SE.}}

\author{Piergiacomo \textsc{Sabino}\footnote{piergiacomo.sabino@eon.com}\\
Quantitative Methods, E.ON SE\\
\vspace{5pt}
 Br\"{u}sseler Platz 1, 45131 Essen, Germany
}

\date{}
\begin{document}
    \maketitle \thispagestyle{empty}
        \begin{abstract}
\noindent In this study we define a three-step procedure to relate the self-decomposability of the stationary law of a generalized Ornstein-Uhlenbeck process to the law of the increments of such processes.			

Based on this procedure and the results of Qu et al.\mycite{QDZ19}, we derive the exact simulation, without numerical inversion,
of the skeleton of a Variance Gamma, and of a symmetric Variance Gamma driven Ornstein-Uhlenbeck process. 
Extensive numerical experiments are reported to demonstrate the accuracy and efficiency of our algorithms. 

These results are instrumental to simulate the spot price dynamics in
energy markets and to price Asian options and gas storages by Monte Carlo simulations in a framework similar to the one discussed in Cummins et al.\mycite{CKM17, CKM18}. 

\vspace{0.2cm}
				\end{abstract}

\section{Introduction}\label{sec:Introduction}
The modeling based on non-Gaussian Ornstein-Uhlenbeck (OU) processes has received a considerable attention in the recent
literature in an attempt to accommodate features such as jumps, heavy tails and asymmetry
which are well evident in real phenomena. For instance, with regards to financial and econometric applications, energy markets, and commodity markets in general, exhibit mean-reversion, seasonality and sudden spikes; mean-reversion in particular, cannot be captured by ordinary \Levy\ processes.

The availability of simulation techniques of easy implementation is important for analysis,
validation and estimation purposes. Indeed, direct likelihood analysis is often
impracticable for these models, whereas, Monte Carlo (MC) based techniques and generalized method of moments (GMM) approaches can be a viable route to estimate the model parameters.

As observed in a series of papers by Barndorff-Nielsen\mycite{BN98}, Barndorff-Nielsen et al.\mycite{BJS1998}, Barndorff-Nielsen and Shephard\mycite{BNSh01, BNSh03}, the concept of self-decomposability (see Sato\mycite{Sato} and Cufaro Petroni\mycite{Cufaro08}) plays an essential role in the theory of generalized OU-processes. In this paper we  define a simple three-steps procedure to determine the characteristic function, and therefore the cumulant function, of a OU process and its relation to the characteristic function of what we name the \emph{a-reminder} of a self-decomposable law. 
This machinery demonstrates to be a powerful tool to determine simulation algorithms for generalized OU-processes as well as to simplify already existing proofs (see for instance Qu et al.\mycite{QDZ19} and Bianchi et al.\mycite{Bianchi2017}).  

Relying on the results of Qu et al.\mycite{QDZ19}, the main contribution of this article is the development of exact simulation schemes to generate the skeleton of Variance Gamma (VG) driven OU processes (OU-VG) discussed in Cummins et al.\mycite{CKM17, CKM18}. The extensive simulation experiments show that our algorithms are efficient and accurate therefore, suitable for concrete applications.

To this end, the modeling of energy markets with non-Gaussian OU processes has been 
discussed, among others, in Benth et al.\mycite{BKM07}, Meyer-Brandis and Tankov\mycite{MBT2008}
and recently in Benth and  Pircalabu\mycite{BenthPircalabu18} in the context of modeling wind power
futures. Compared to mean-reverting jump-diffusion models (see for instance, Cartea and Figueroa\mycite{CarteaFigueroa} and Kjaer\mycite{Kjaer2008})   these models exhibit the competitive advantage of having less parameters.

We illustrate the applicability of our schemes in the pricing Asian options and gas storages by MC simulation using market dynamics similar to those discussed in Cummins et al.\mycite{CKM17, CKM18}. Once more, our algorithms demonstrate to be efficient and reasonably fast to compute the fair values of such energy derivatives. Although MC methods are not as fast as other numerical techniques as FFT and quantization methods (see for instance Jaimungal and Surkov\mycite{JaimungalSurkov11} and Bardeau et al. \mycite{BBP07}), they nevertheless, give the possibility to compute different quantiles of the price distribution and are independent on contract payoffs. 

The remainder of the article is organized as follows. Section\myref{sec:Preliminaries} recalls the properties of generalized OU processes 
and introduces the conceptual procedure which we will use in order to develop the simulation
schemes. In Section\myref{sec:OUVG} we derive the characteristic function of the law of the increments of OU-VG processes and develop simulation schemes for the skeleton of such processes. In this section, we also demonstrate the effectiveness of our algorithms through extensive numerical experiments. Section\myref{sec:FinancialApplications} illustrates some financial applications: we consider the pricing of Asian options by MC simulations in a $2$-factor market driven by the sum of a standard VG process and a OU-VG process then, we consider the pricing of gas storages using a one-factor spot dynamics similar to the setting discussed in Cummins et al.\mycite{CKM17, CKM18}. Finally, Section\myref{sec:Conclusions} concludes the paper with an overview of future inquiries and further possible applications. 
\section{Preliminaries}\label{sec:Preliminaries}
Following Barndorff-Nielsen and Shephard\mycite{BNSh01}, we consider a \Levy\ process $Z(t)$ and  the
generalized OU process defined by the \sde
            \begin{equation}\label{eq:genOU_sde}
              dX(t) =  -kX(t)dt + dZ(t) \quad\qquad X(0)=X_0\quad \Pqo\qquad
              k>0.
            \end{equation}
with solution 
\begin{equation}
X(t) = X(0)\,e^{-kt} + \int_0^t e^{k(t-v)}dZ(v).
\label{eq:sol:OU}
\end{equation}
 Here $Z(t)$ is called the Backward Driving \Levy\ Process
(\emph{BDLP}), and we will adopt the following notation: if
$\mathfrak{D}$ is the stationary law of $X(t)$, we will say that
$X(t)$ is a $\mathfrak{D}$-\ou\ process; if on the other hand, $Z(1)$
  is distributed
according to the  \id\ (infinitely divisible) law
$\widetilde{\mathfrak{D}}$, then we will say that $X(t)$ is an
\ou-$\widetilde{\mathfrak{D}}$ process. Now a well known result (see
for instance Cont and Tankov~\cite{ContTankov2004} or
Sato~\cite{Sato}) is that,  a given one-dimensional
distribution $\mathfrak{D}$ always is the stationary law of a
suitable \ou-$\widetilde{\mathfrak{D}}$ process if and only if
$\mathfrak{D}$ is self-decomposable.

We recall that a law with probability density (\pdf) $f(x)$ and
characteristic function (\chf) $\varphi(u)$ is said to be
\emph{self-decomposable} (\sd) (see Sato\mycite{Sato} or Cufaro
Petroni~\cite{Cufaro08}) when for every $0<a<1$ we can find another
law with \pdf\ $g_a(x)$ and \chf\ $\chi_a(u)$ such that
                \begin{equation}\label{aremchf}
                    \varphi(u)=\varphi(au)\chi_a(u)
                \end{equation}
We will accordingly say that a random variable (\rv) $X$ with \pdf\
$f(x)$ and \chf\ $\varphi(u)$ is \sd\ when its law is \sd: looking
at the definition, this means that for every $0<a<1$ we can always
find two \emph{independent} \rv's, $Y$ (with the same law of $X$)
and $Z_a$ (here called \emph{\arem}), with \pdf\ $g_a(x)$
and \chf\ $\chi_a(u)$ such that
                \begin{equation}\label{sdec-rv}
                    X\eqd aY+Z_a\qquad\quad\Pqo
                \end{equation}
As observed in Barndorff\mycite{BJS1998}, $X(t)$ is stationary if and only if the \chf\ of the \sd\ stationary law $\varphi_X(u)$ is of the form $\varphi_X(u)=\varphi_X(u\,e^{-kt})\chi(u,t)$, where $\chi(u,t)$ denotes the \chf\ of the second term of Equation\refeq{eq:sol:OU}. Defining the cumulant function of a \rv\ $Y$ as $\kappa_Y(u)=\log\EXP{e^{uY}}$, it turns out  that there is precise relation between the  cumulant function of the stationary distribution $\bar{\kappa}_X(u)$, that of $Z(1)$, denoted $\kappa_Z(u)$, and that of the second term of Equation\refeq{eq:sol:OU}, denoted $\varrho_X(u,t)$ (see also Taufer and Leonenko\mycite{TAUFER2009} and Schoutens\mycite{Schoutens03}).
\begin{eqnarray}
\bar{\kappa}_X(u) &=& \int_0^{+\infty} \kappa_Z(u\,e^{-kt}) ds\\
\varrho_X(u, t) &=& \bar{\kappa}_X(u) - \bar{\kappa}_X(u\,e^{-kt})   
\label{eq:ou:cumulants}
\end{eqnarray}
The last equation means that the law of the second term of Equation\refeq{eq:sol:OU} coincides with that of the \arem\ of the law of the stationary distribution if one takes $a=e^{-kt}$. 
A similar observation was also mentioned in Gaver and Lewis\mycite{gaver_lewis_1980}, Lawrence\mycite{L80} and later in Wolfe\mycite{WOLFE1982} in the context of first order auto-regressive processes $X_n = \rho X_{n-1} + \epsilon_n$, in case $0\le \rho <1$, that are the discrete-time equivalent of OU processes.

This facts give a useful machinery to determine the \chf\ or the cumulant of OU processes which of course, can be used to find simulation algorithms.
\begin{itemize}
	\item Find the cumulant function of the stationary distribution given the BDLP.
	\item Find the cumulant function of the \arem\ of the marginal distribution.
	\item Set $a = e^{-kt}$.
\end{itemize}
On the other hand, based on the observations above, the sequential generation of the skeleton of $X(t)$ on a time grid $t_1, \dots t_M$
consists in finding a simulation algorithm for the \arem\ of the stationary law assuming at each step $a_i=e^{-k(t_{i}-t_{i-1})}, i=1, \dots, M$.
Hereafter, without loss of generality, we will assume an equally-spaced time grid with $\Delta t = t_i-t_{i-1}, \forall i=1,\dots, M$.

Finally, because the cumulant function $\kappa_X(u,t)=\log\EXP{e^{uX(t)}}$ can also be written in terms of the cumulant function $\kappa_Z(u)$ of $Z(1)$ as (see Cont and Tankov\mycite{ContTankov2004} Lemma 15.1)
\begin{equation}
	\kappa_X(u,t) = uX(0)e^{-kt} + \int_0^t\kappa_Z\left(ue^{-k(t-v)}\right)dv=uX(0)e^{-kt} + \varrho_X(u,t),
\label{eq:cumulant:function}
\end{equation}
one can relate the cumulants $\kappa_{X,n}$ of $X(t)$ to the cumulants $\kappa_{Z,n}$ of $Z(1)$
\begin{eqnarray}
\EXP{X(t)} &=& X(0)e^{-kt} + \frac{\kappa_{Z,1}}{k}\left(1-e^{-kt}\right)\label{eq:cumulants:ou1}\\
\kappa_{X,n} &=& \frac{\kappa_{Z, n}}{n\,k}\left(1-e^{-nkt}\right), \quad n\ge 2
\label{eq:cumulants:ou2}
\end{eqnarray}
and therefore one can have useful benchmarks to test the performance of the simulation algorithms or to carry out an estimation procedure based on the generalized method of moments.

\section{OU-VG Processes and their Exact Simulation}\label{sec:OUVG}
The VG process, introduced in Madan and Seneta\mycite{MadanSeneta90}, can be seen as a Brownian Motion (BM) where the clock ticks with a random time described by a gamma subordinator $G(t)$.

We recall that the gamma law $\Gamma(\alpha, \beta)$ is famously \sd\ (see Grigelionis\cite{Gri03}) and has the following \pdf\ and \chf\  
	\begin{eqnarray*}
			f_{\alpha, \beta}(x)&=&\frac{\beta^{\alpha}x^{\alpha-1}}{\Gamma(\alpha)}e^{- \beta x}\mathbbm{1}_{x \ge 0}\\
			\phi_{\Gamma}(u) &=& \left(\frac{\beta}{\beta - iu}\right)^{\alpha}
	\end{eqnarray*}
	where $\Gamma(\cdot)$ is the Euler gamma function,  $\alpha>0$ and $\beta>0$ are called \emph{shape} and \emph{rate} parameters, respectively.
	
	The gamma process $G(t, \alpha, \beta)$ is a continuous-time process 
	with stationary, independent gamma increments such that
	for any $h > 0$,
	\begin{equation}
			G(t + h, \alpha, \beta)-G(t, \alpha, \beta) \sim \Gamma(\alpha h, \beta).
	\end{equation}
	therefore, it is a subordinator (see Sato\mycite{Sato}).
        In order to guarantee that the stochastic clock $G(t)$ is an unbiased reflection of calendar time (see Joshi\mycite{Joshi2005}) we need to set $\EXP{G(t)}=t$. The law of the increments of $G(t)$ now depends on one parameter only $\nu=\frac{1}{\alpha}= \frac{1}{\beta}$.
Denoting now $G(t, \nu)$ the gamma process with the above parameters restriction, the VG process is defined as follows:
\begin{equation}\label{eq:def:dVG}
	V(t) = \theta G(t, \nu) + \sigma W(G(t, \nu)),
\end{equation}
with characteristic exponent (\che) $\psi_{VG}(u)$:
\begin{equation}\label{eq:che:vg}
	 \psi_{VG}(u) = \log\EXP{e^{iuV(1)}}=-\frac{1}{\nu}\log\left(1 - iu\theta\nu + u^2\frac{\sigma^2\nu}{2}\right).
\end{equation}
for $\sigma>0$ and $\theta \in \R$ constants.
				
Since the VG is a process of finite variation, it can be written
as difference of two increasing gamma processes
\begin{equation*}
V(t) = \gamma_p(t, \mu_p, \nu_p) - \gamma_n(t, \mu_n, \nu_n)
\end{equation*} 
with $\mu_p, \nu_p, \mu_n, \nu_u, \nu$ satisfying the following equations
\begin{eqnarray*}
\mu_p &=& \frac{1}{2}\sqrt{\theta^2 + \frac{2\sigma^2}{\nu}} + \frac{\theta}{2}, \\
\mu_n &=& \frac{1}{2}\sqrt{\theta^2 + \frac{2\sigma^2}{\nu}} - \frac{\theta}{2}, \\
\nu_p &=& \mu_p^2 \nu, \\
\nu_n &=& \mu_n^2 \nu.	
\end{eqnarray*}
The \che\ can then be rewritten as
\begin{equation}
	\psi_{VG}(u) = -\frac{1}{\nu}\log\left(1-iu\frac{\nu_p}{\mu_p}\right) -\frac{1}{\nu}\log\left(1+iu\frac{\nu_n}{\mu_n}\right)=\psi_{\Gamma_p}(u)+\psi_{\Gamma_n}(-u),
	\label{eq:vg:ch:exponent}
\end{equation}
where $\psi_{\Gamma_p}(u)$ and $\psi_{\Gamma_n}(u)$ are the \che's of a $\Gamma(\frac{1}{\nu}, \frac{\mu_p}{\nu_p})$ and a $\Gamma(\frac{1}{\nu}, \frac{\mu_n}{\nu_n})$ law, respectively (therefore of the difference of two independent gamma-distributed \rv's). 
When $\theta=0$ - in case of a symmetric VG (SVG) - it simplifies to
\begin{equation}
	\psi_{SVG}(u) = -\frac{1}{\nu}\log\left(1 + u^2\frac{\sigma^2\nu}{2}\right) = \psi_{\Gamma}(u) - \psi_{\Gamma}(-u).
	\label{eq:svg:ch:exponent}
\end{equation}
where now $\psi_{\Gamma}(u)$ denotes the \che\ of a $\Gamma(\frac{1}{v}, \frac{2}{\sigma^2\nu})$ law.
Using the machinery illustrated in Section~\ref{sec:Preliminaries}, we can calculate the cumulant function $\bar{\kappa}_X(u)$ of the stationary distribution of the VG driven OU process 
\begin{equation}
	X(t) = X(0)e^{-kt} + \int_0^te^{-k(t-v)}dV(v)
\label{eq:OUVG}
\end{equation}
and the cumulant function $\varrho_X(u,t)$.
\begin{prop}
	The cumulant function $\bar{\kappa}_X(u)$ and $\varrho_X(u,t)$ of a OU-VG process are given by
	\begin{equation}
		\bar{\kappa}_X(u)= \frac{1}{k\nu}\left(\DiLog{u\frac{\mu_p}{\nu_p}}+ \DiLog{-u\frac{\mu_n}{\nu_n}}\right).
	\label{eq:cumulant:OUVG:stat}
	\end{equation}
	\begin{equation}
		\varrho_X(u,t) = \frac{1}{k\nu}\left(\DiLog{u\frac{\mu_p}{\nu_p}}-\DiLog{u\frac{\mu_p\,e^{-kt}}{\nu_p}} + \DiLog{-u\frac{\mu_n}{\nu_n}}+\DiLog{-u\frac{\mu_n\,e^{-kt}}{\nu_n}}\right)
	\label{eq:cumulant:OUVG}
	\end{equation}
\end{prop}
\begin{proof}
Denoting $\kappa_{VG}(u)=\psi_{VG}(-iu)$, we have
\begin{equation*}
	\bar{\kappa}_X(u)=\int_0^{+\infty} \kappa_{VG}(ue^{-ks})ds = \int_0^{+\infty} \kappa_{\Gamma_p}(u
	e^{-ks})ds + \int_0^{+\infty} \kappa_{\Gamma_n}(-ue^{-ks})ds
\end{equation*}
where $\kappa_{\Gamma_p}(u)=\psi_{\Gamma_p}(-iu)$ and $\kappa_{\Gamma_n}(u)=\psi_{\Gamma_n}(-iu)$.
With the change of variable $x=e^{-ks}$ we have
\begin{equation*}
	\bar{\kappa}_X(u)= -\frac{1}{k\nu}\left(\int_0^1 \frac{\log(1-u\,x\frac{\nu_p}{\mu_p})}{x}dx + 
	\int_0^1 \frac{\log(1+u\,x\frac{\nu_n}{\mu_n})}{x} dx\right)
\end{equation*}	
Of course, the last two terms are the cumulant functions of a stationary $OU-\Gamma$ processes and of its negative counterpart (see Qu et al.\mycite{QDZ19} and Table 2 in Barndorff-Nielsen and Shephard\mycite{BNSh03}) whichx can be written in terms of the dilogarithmic Spencer's function $\dilog(z)=-\int_0^z\frac{\log(1-y)}{y}dy, z\in\mathbb{C}$ (see Gradshteyn and Ryzhik\mycite{gradshteyn2007})
\begin{equation*}
	\bar{\kappa}_X(u)= \frac{1}{k\nu}\left(\DiLog{u\frac{\mu_p}{\nu_p}} + \DiLog{-u\frac{\mu_n}{\nu_n}}\right).
\end{equation*}
Hence
\begin{eqnarray*}
	\varrho_X(u,t) &=& -\frac{1}{k\nu}\left(\int_{e^{-kt}}^1 \frac{\log\left(1-u\,x\frac{\nu_p}{\mu_p}\right)}{x}dx + 
	\int_{e^{-kt}}^1 \frac{\log\left(1+u\,x\frac{\nu_n}{\mu_n}\right)}{x} dx\right)\\
	&=& \frac{1}{k\nu}\left(\DiLog{u\frac{\mu_p}{\nu_p}}-\DiLog{u\frac{\mu_p\,e^{-kt}}{\nu_p}} + \DiLog{-u\frac{\mu_n}{\nu_n}}-\DiLog{-u\frac{\mu_n\,e^{-kt}}{\nu_n}}\right)
\end{eqnarray*}
that concludes the proof.	
\end{proof}
By simply setting $\theta=0$ we retrieve the cumulant function relative to the OU-SVG process of Cummins et al.\mycite{CKM17}.
	\begin{equation}
		\bar{\kappa}_X(u)= \frac{1}{2k\nu}\dilog\left(u^2\frac{\sigma^2\nu}{2}\right),
	\label{eq:cumulant:OUSVG:stat}
	\end{equation}
	and
	\begin{equation}
		\varrho_X(u,t) = \frac{1}{2k\nu}\left(\dilog\left(u^2\frac{\sigma^2\nu}{2}\right) - \dilog\left(u^2\frac{\sigma^2\nu}{2}\,e^{-kt}\right)\right)
	\label{eq:cumulant:OUSVG}
	\end{equation}


\subsection{Simulation Algorithms}\label{subsec:SimulationAlgorithms}
From the results of the previous section we can conclude that the simulation of a OU-VG process consists in the repetition of the simulation a OU-$\Gamma$ process two times and then take the difference. To this end, Qu et al.\mycite{QDZ19} found that a \rv\ $Y$ with cumulant function 
\begin{equation*}
	\varrho_Y(u, \Delta t) = -\frac{\alpha}{k}\int_{e^{-k\Delta t}}^1 \frac{\log\left(1+\frac{u\,x}{\beta}\right)}{x}dx
\end{equation*}	
can be decomposed into the sum of a gamma-distributed \rv\ $Y_1\sim\Gamma(\alpha\,\Delta t, \beta\,e^{k\Delta t})$ and a compound Poisson process $Y_2=\sum_{m=1}^MJ_m$ with intensity $\lambda=\frac{\alpha k\Delta t^2}{2}$ and exponentially-distributed jumps $J_m$ with random rate $\beta\,e^{k\Delta t\sqrt{U}}$ and $U\sim\mathcal{U}([0,1])$. Based on this result, the simulation of skeleton of an OU-VG on an equally-spaced time grid $t_0, t_1, \dots, t_M$ with step $\Delta t$ consists in nothing less than simulating a $Y$-like \rv\ two times at each step and then taking the difference  as illustrated in Algorithm~\ref{alg:sim:vg}.
\begin{algorithm}
\caption{ }\label{alg:sim:vg}
\begin{algorithmic}[1]
		\For{ $m=1, \dots, M$}
		\State Generate $G_p\sim\Gamma(\frac{\Delta t}{\nu}, \frac{\mu_p}{\nu_p}\,e^{k\Delta t})$ 
		\State Generate $G_n\sim\Gamma(\frac{\Delta t}{\nu}, \frac{\mu_n}{\nu_n}\,e^{k\Delta t})$. 
		\State Generate $r\sim\poiss(\frac{k\Delta t^2}{2\nu})$\Comment{Poisson \rv\ with intensity $\frac{k\Delta t^2}{2\nu}$}
		\State Generate $s\sim\poiss(\frac{k\Delta t^2}{2\nu})$
		\State Generate $r$ \iid\ uniform \rv's $\bm{u}=(u_1, \dots, u_r)\,\sim\,\unif([0,1]^r)$.
		\State Generate $s$ \iid\ uniform \rv's $\bm{v}=(v_1, \dots, v_s)\,\sim\,\unif([0,1]^s)$.
		\State $\beta_{p,i}\gets\frac{\mu_p}{\nu_p} e^{k\Delta t \sqrt{u_i}}, i=1,\dots, r$.
		\State $\beta_{n,j}\gets\frac{\mu_n}{\nu_n} e^{k\Delta t \sqrt{v_j}}, j=1,\dots, s$.
		\State Generate $r$ \iid\ $J_{p,i}\sim\erl_1(\beta_{p,i}), i=1,\dots, r$, \Comment {Exponential \rv's with rate $\beta_{p,i}$}
		\State Generate $s$ \iid\ $J_{n,j}\sim\erl_1(\beta_{n,j}), j=1,\dots, s$, \Comment {Exponential \rv's with rate $\beta_{n,j}$}
		\State $X(t_m)\gets X(t_{m-1})e^{-k\Delta t} + G_p - G_n + \sum_{i=1}^r J_{p,i} - \sum_{j=1}^s J_{n,j}$.
		\EndFor
		\end{algorithmic}
\end{algorithm}
On the other hand, the simulation procedure of a symmetric VG can be simplified observing that for $\theta=0$, $\frac{\mu_p}{\nu_p}=\frac{\mu_n}{\nu_n}=\frac{1}{\sigma}\sqrt{\frac{2}{\nu}}$ and that the difference $C=C_1-C_2$ of two \iid\ compound Poisson processes $C_1=\sum_{n_1}^{N_1}U_{n_1}$ and $C_2=\sum_{n_2}^{N_2}D_{n_2}$ with intensity $\lambda$ and with exponentially distributed jumps has the same law of $\sum_n^N(U_n-D_n)$ with intensity $2\lambda$. It is well known that the difference of exponentially distributed \rv's with the same scale parameter $\mu$ is distributed according to a central Laplace law $\lap(\mu)$. Such a \rv\ can be efficiently generated using the inverse transformation method (see Devroye\mycite{Dev86}) observing that the inverse of the cumulative distribution is
\begin{equation*}
	F_L^{-1}(y) = -\mu\sign(y-0.5)\,\ln(1 - 2|y-0.5|).
\end{equation*} 
Based on these observations, the steps of the sequential simulation of the skeleton of a symmetric OU-SVG process are summarized in Algorithm~\ref{alg:sim:svg}
\begin{algorithm}
\caption{ }\label{alg:sim:svg}
\begin{algorithmic}[1]
		\For{ $m=1, \dots, M$}
		\State Generate $G_p\sim\Gamma(\frac{\Delta t}{\nu}, \frac{1}{\sigma}\sqrt{\frac{2}{\nu}}\,e^{k\Delta t})$ 
		\State Generate $G_n\sim\Gamma(\frac{\Delta t}{\nu}, \frac{1}{\sigma}\sqrt{\frac{2}{\nu}}\,e^{k\Delta t})$. 
		\State Generate $r\sim\poiss(\frac{k\Delta t^2}{\nu})$\Comment{Poisson \rv\ with intensity $\frac{k\Delta t^2}{\nu}$}
		\State Generate $r$ \iid\ uniform \rv's $\bm{u}=(u_1, \dots, u_r)\,\sim\,\unif([0,1]^r)$.
		\State $\mu_{i}\gets\sigma\sqrt{\frac{\nu}{2}} e^{-k\Delta t \sqrt{u_i}}, i=1,\dots, r$.
		\State Generate $r$ \iid\ $J_{i}\sim\lap(\mu_{i}), i=1,\dots, r$, \Comment {Laplace \rv's with parameter $\mu_{i}$}
		\State $X(t_m)\gets X(t_{m-1})e^{-k\Delta t} + G_p - G_n + \sum_{i=1}^r J_{i}$.
		\EndFor
		\end{algorithmic}
\end{algorithm}


\subsection{Numerical Experiments}\label{subsec:NumericalExperiments}
In this section, we illustrate the performance and effectiveness of our algorithms through extensive
numerical experiments. 
All the simulation experiments in the present paper
have been conducted using \emph{MATLAB R2019a} with a $64$-bit
Intel Core i5-6300U CPU, 8GB 
\footnote{The relative codes are available at \url{https://github.com/piergiacomo75/OUVarianceGamma} }.  
As an additional validation, the
comparisons of the simulation computational times have
also been performed with \emph{R} and \emph{Python}  leading to the same conclusions.

The numerical validation and tests for
our algorithms are based on the comparison to the true expected value, variance, skewness and kurtosis of the OU-VG process. Because of Equations\refeq{eq:cumulants:ou1}and\refeq{eq:cumulants:ou2} and the relation between the cumulants, these quantities relative to $X(t+\Delta t) = X(t)e^{-k\Delta t} + \int_0^{\Delta t}e^{-k(t-v)}V(v)$ are
\begin{eqnarray*}
	\EXP{X(t+\Delta t)}&=& a\, X(t) + (1-a) \frac{\theta}{k}\\
	\VAR{X(t+\Delta t)}&=& (1-a^2)  \frac{\sigma^2 + \theta^2\nu}{2k}\\
	\SK{X(t+\Delta t)}&=& \frac{2\sqrt{2 k}}{3}\frac{(1-a^3)}{(1-a^2)^{3/2}}\frac{2\theta^3\nu^2 + 3\sigma^2\theta\nu}{\left(\sigma^2+\theta^2\nu\right)^{3/2}}\\
	\KUR{X(t+\Delta t)}&=& k\frac{(1+a^2)}{(1-a^2)}\times\frac{3\sigma^4\nu+12\sigma^2\theta^2\nu^2+6\theta^4\nu^3}{\left(\sigma^2 + \theta^2\nu\right)^2} + 3
\end{eqnarray*}
\noindent where $a=e^{-k\Delta t}$ whereas, the symmetric case is simply obtained with $\theta=0$. 

Table~\ref{tab:ou:vg:MC} reports the CPU times in seconds and compares the MC estimated values of the true $\EXP{X(T)}$, $\VAR{X(T)}$, $\SK{X(T)}$ and $\KUR{X(T)}$. The values at the top of the table are obtained with a single time step $T=\Delta t=1/5$ whereas, those at the bottom are relative to a time grid of five points once more with step $1/5$. Varying the number of simulations $N_S$, we can conclude that our algorithm is efficient and convergent, although it seems that at least $N_S=10^4$ simulations is required to achieve a good estimate.  However, although the algorithm provides an exact simulation of a OU-VG process, the generation of an entire trajectory, especially over a time grid with several points is not extremely fast compared to the simulation of other OU processes (see for instance Cufaro Petroni and Sabino\mycite{cs20_2, cs20_1}). Figure~\ref{fig:ou:vg} shows a sample trajectory using a time grid of $365$ points that is a quite common choice in financial applications relative to the pricing of a one year contract.

We conclude this section illustrating the results of the numerical experiments relative to a OU-SVG process. 
$\EXP{X(T)}=X(0) e^{-kT}$ and the skewness is zero therefore, in Table~\ref{tab:ou:sym:vg:MC} we show the CPU times in seconds and the MC estimated values of the true $\VAR{X(T)}$ and $\KUR{X(T)}$ only. The simulation has been conducted using the same time grid as the previous case whereas, the process parameters are those presented in Cummins et al.\mycite{CKM17}; Figure~\ref{fig:ou:sym:vg} shows a sample path of such a processes over a time grid of $365$ points.
Once again, the simulation algorithm is convergent and captures the true values of the benchmarks quite well. Moreover, Algorithm~\ref{alg:sim:svg} is faster than Algorithm~\ref{alg:sim:vg} of almost a factor $2$ and provides an efficient solution when the drift $\theta$ can be neglected. 

\begin{figure}
\caption{Trajectories with $M=365$, $\Delta t =1/365$.}\label{fig:trajectories}
		\begin{subfigure}[c]{.5\textwidth}{
				\includegraphics[width=70mm]{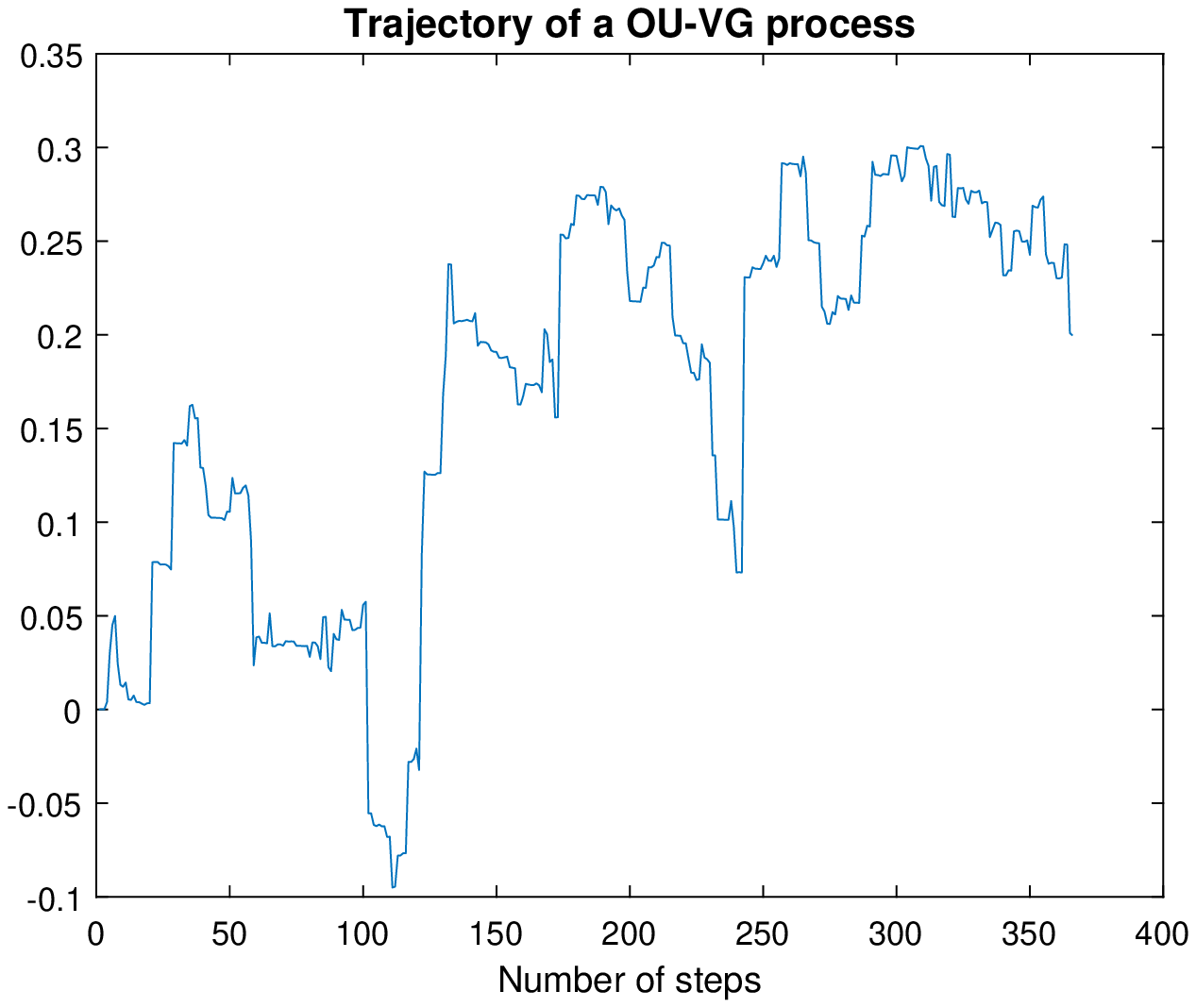}
				}
				\caption{OU-VG $X(0)=0$, $\theta=0.025$, $k=0.2$, $\nu=0.02$, $\sigma=0.3$.}\label{fig:ou:vg}
		\end{subfigure}
		\begin{subfigure}[c]{.5\textwidth}{
				\includegraphics[width=70mm]{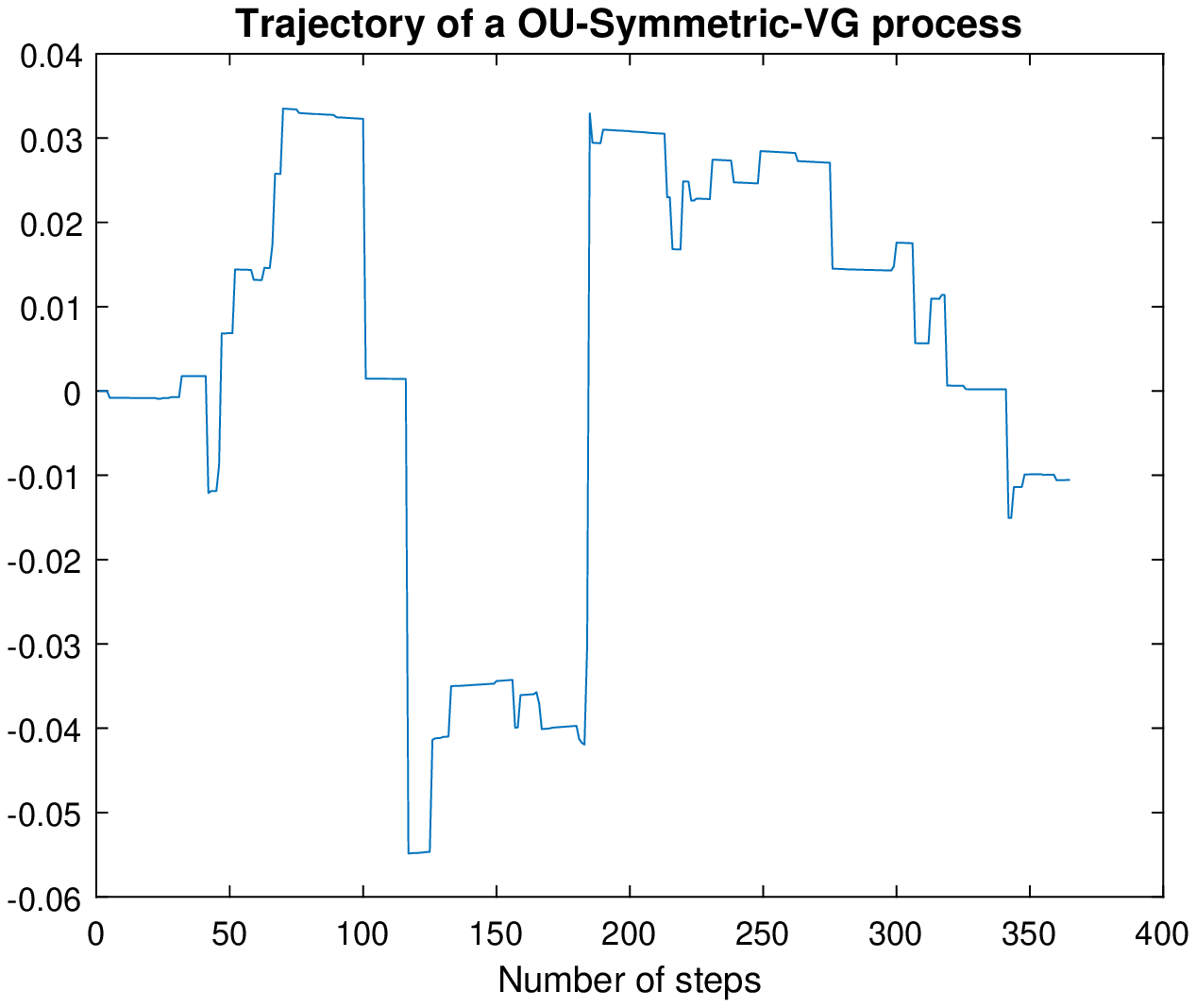}
				}
				\caption{OU-SVG with $X(0)=0$, $k=0.2162$, $\nu=0.256$, $\sigma=0.201$.}\label{fig:ou:sym:vg}
		\end{subfigure}
\end{figure}%

\input{./Tables/OU_VG}

\input{./Tables/OU_Sym_VG}


\section{Financial Applications}\label{sec:FinancialApplications}
Pricing derivative contracts or energy facilities is often accomplished using MC methods; for this purpose is therefore, necessary to rely on efficient and eventually, fast path-generation techniques. On the other hand, the day-ahead (also called spot) price of power or gas and in general of commodities exhibit mean-reversion, seasonality and spikes, this last feature is particularly difficult to be captured in a pure Gaussian world. Different approaches have been investigated in order to somehow extend the classical Gaussian framework introduced in Lucia and Schwarz\mycite{LS02} and Schwartz and Smith\mycite{SchwSchm00}. Among others, Cartea and Figueroa\mycite{CarteaFigueroa}, Kjaer\mycite{Kjaer2008}, Meyer-Brandis  and P. Tankov\mycite{MBT2008} have studied mean-reverting jump-diffusions to model sudden spikes, whereas, Benth et al.\mycite{BKM07} and Benth and A. Pircalabu\mycite{BenthPircalabu18} have considered different non-Gaussian OU processes in order to price power or wind derivative contracts. Recently, Cummins et al.\mycite{CKM17, CKM18} have addresses the pricing of gas storages via FFT in a market driven by a OU-SVG process. 

In the following subsections we illustrate the effectiveness of the algorithms presented in subsection~\ref{subsec:SimulationAlgorithms} when applied to the pricing of Asian options and of gas storages using market models similar to those considered in Cummins et al.\mycite{CKM17, CKM18}. Of course, MC methods are known to be sometimes slower than FFT and other techniques, nevertheless, they provide a view on the distribution of the potential cash-flows of derivative contracts giving a precious information to risk managers or to trading units. The calibration and in general, the parameter estimation of OU-VG processes is not the aim of the paper. However, as observed in Wolfe\mycite{WOLFE1982}, discrete first-order autoregressive processes are embedded into continuous OU processes, therefore one could use the generalized method of moments (GMM) to derive Yule-Walker-like equations and estimate the model parameters from historical data.   
  
\subsection{Asian Options\label{subsect:Asian}}
In this section we assume that the spot price of a gas market is driven by the following $2$-factors process
	\begin{equation}\label{eq:ouvg_vg_market}
		S(t) = F(0,t)\, e^{h(t) + X_1(t) + X_2(t)}=F(0,t)\, e^{h(t) + H(t)}
	\end{equation}
	where $h(t)$ is a deterministic function, $F(0,t)$ is the forward curve and $X_1(t)$ is a OU-VG process with parameters ($k, \theta_1, \nu_1, \sigma_1$). In contrast to Cummins et al.\mycite{CKM18}, we add a second independent VG process $X_2(t)$ with parameters $(\theta_2, \nu_2, \sigma_2$) to capture the long-term behavior.  
	Using the risk-neutral arguments  of the Lemma 3.1 in Hambly
et al.\mycite{HHM11},  the deterministic function $h(t)$
 consistent with forward curve is
\begin{equation}\label{eq:rn:spot}
	h(t) = -\kappa_H(1, t).
\end{equation}
where $\kappa_H(u, t)$ is the cumulant function of the process $H(t)$ at time $t$, then because of Equations\refeq{eq:che:vg} and\refeq{eq:cumulant:OUVG}
\begin{eqnarray}\label{eq:rn:spot:model1}
	h(t) &=&  -\frac{1}{k\nu}\left(\DiLog{\frac{\mu_p}{\nu_p}}-\DiLog{\frac{\mu_p\,e^{-kt}}{\nu_p}} + 
	\DiLog{-\frac{\mu_n}{\nu_n}}-\DiLog{-\frac{\mu_n\,e^{-kt}}{\nu_n}}\right) - \nonumber\\
	&&\frac{t}{\nu}\log\left(1 - \theta\nu - \frac{\sigma^2\nu}{2}\right),
\end{eqnarray}
where $\mu_p, \nu_p, \mu_n, \nu_n$ are relative to the parameters of $X_1(t)$.

Finally, we recall that the payoff at maturity $T$ of an Asian option with European style and strike price $K$ is
\begin{equation*}
	A(T) = \left(\sum_{i=1}^{d}\omega_iS(t_i) - K\right)^+.
\end{equation*}
In our numerical experiments we assume an at-the-money Asian option $K=F(0,0)=15$ having one year maturity, $T=1$, with equal weights $\omega_i=1/d$ and
with a flat forward curve. Although, as already mentioned, we do not focus on the parameters estimation, the values in Table\myref{tab:Asian:Parameters} could be considered realistic because they are based on the estimations presented in Gardini et al.\mycite{GSS20_1} for $X_2(t)$ and are similar to those shown in Cummins et al.\mycite{CKM17} for $X_1(t)$; Figure\myref{fig:trajectory:asian} shows one sample price-path generated with Algorithm\myref{alg:sim:vg} using the market model of Equation\refeq{eq:ouvg_vg_market}.

\input{./Tables/Asian_Results}
\begin{figure}\label{fig:price:trajectories}
\caption{Price Trajectories}
		\begin{subfigure}[c]{.5\textwidth}{
			\includegraphics[width=70mm]{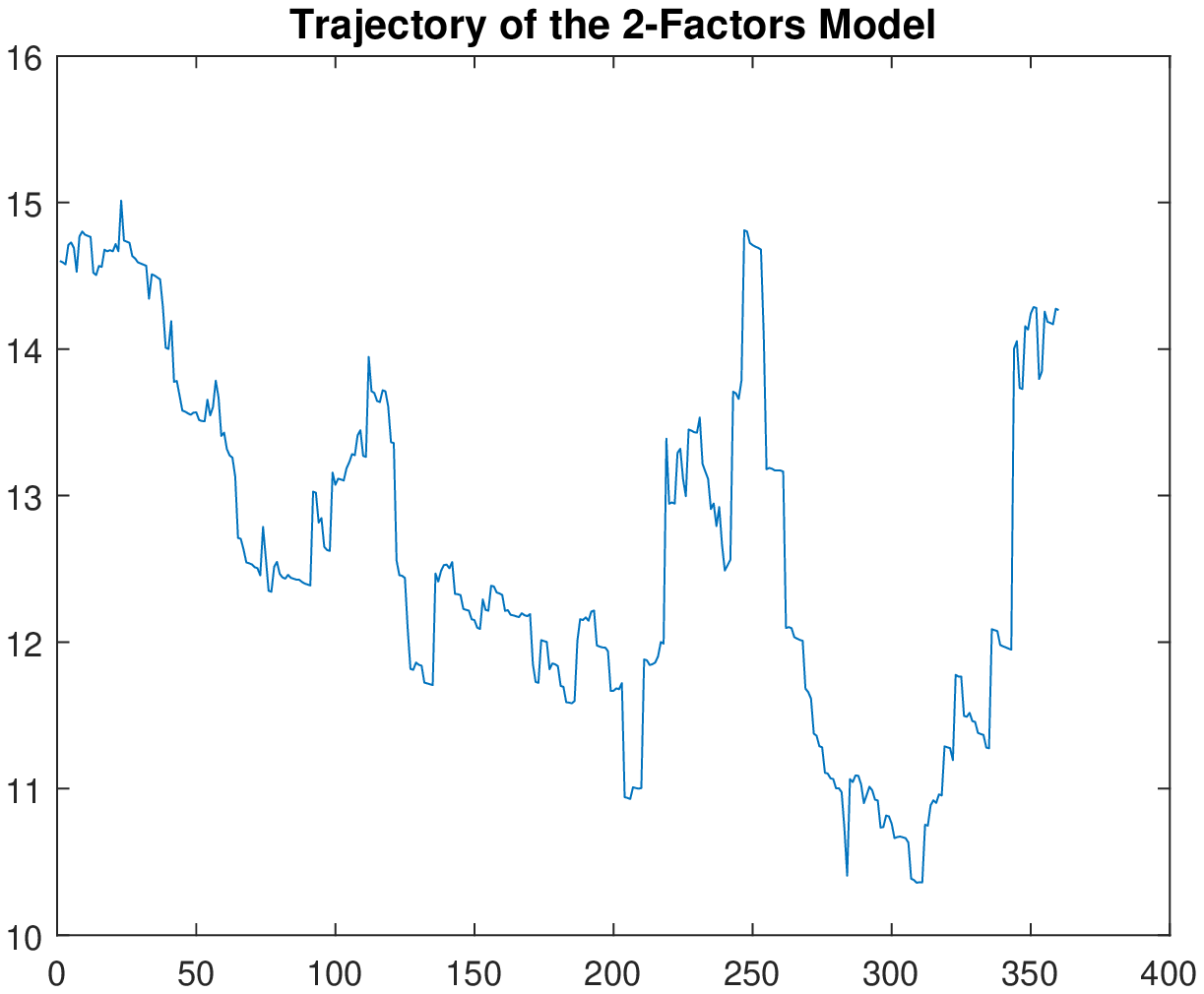}
			\caption{$2$-factors Model}\label{fig:trajectory:asian}
			}
		\end{subfigure}
		\begin{subfigure}[c]{.5\textwidth}{
			\includegraphics[width=70mm]{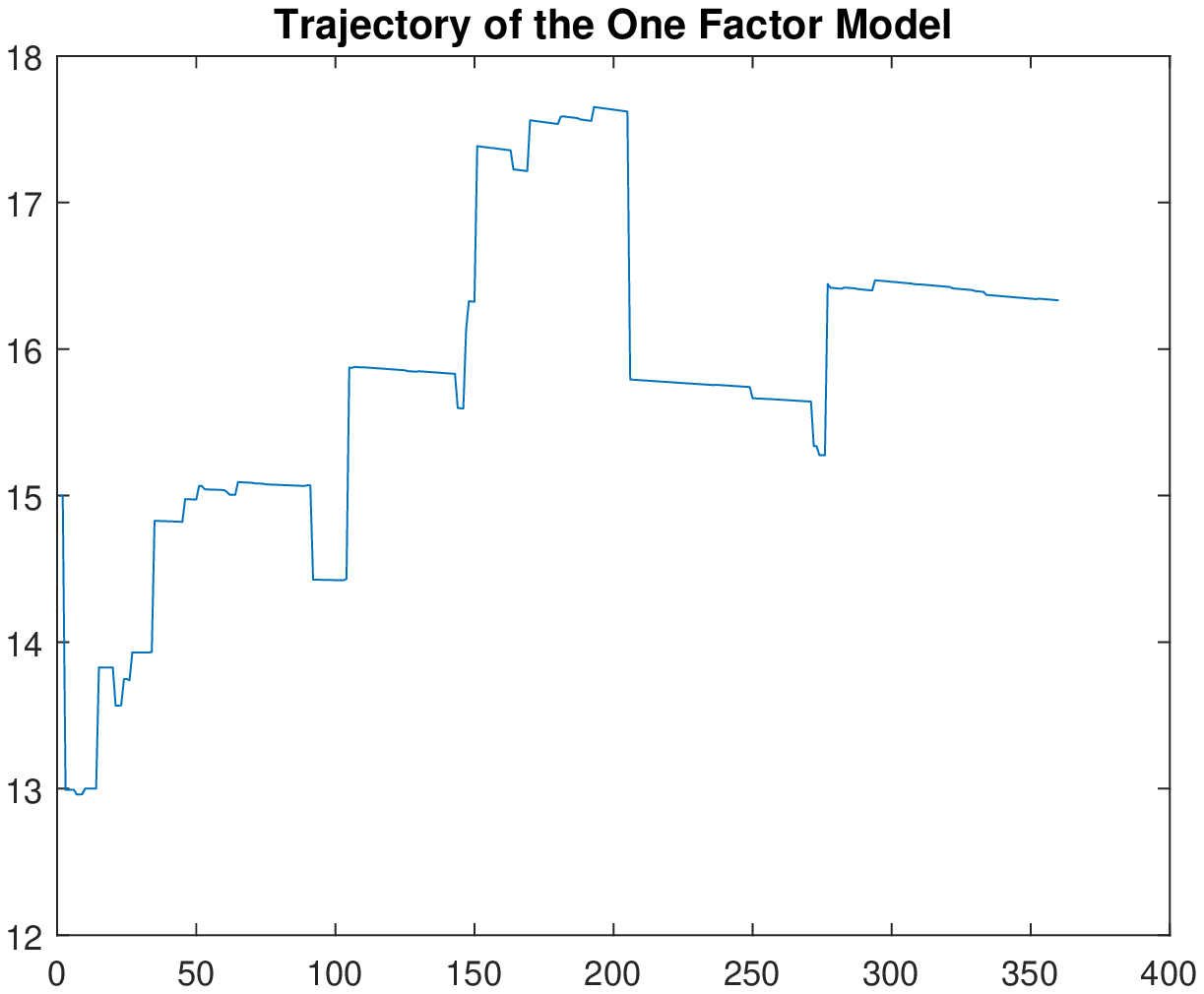}
			\caption{One factor Model}\label{fig:trajectory:gas:storage}
			}
		\end{subfigure}
\end{figure}
Table\myref{tab:Asian:Results} shows the estimated prices obtained by MC varying the number of simulations $N_S$, along with the overall computational times (CPU) in seconds. The columns \emph{stdev} and \emph{error} report the standard deviations of the MC estimator and the errors around the estimated option prices defined as the standard deviation divided by $\sqrt{N_S}$. The results illustrate that our simulation scheme is accurate indeed, the price of the option converges rapidly and the errors are very small; it seems that $10000$ simulations are good enough to have a reliable Asian option price. On the other hand, although here combined with the simulation of a standard VG process, the simulation of a OU-VG process is not extremely fast compared to the one of other generalized OU-processes (see for instance Cufaro Petroni and Sabino\mycite{cs20_2, cs20_1}). Nevertheless, it provides additional information, such as quantiles of the price distribution, that can be employed to derive risk premia and support decision making and in particular, can give an insight whether the calibrated parameters imply realistic price trajectories.

\subsection{Gas Storages}\label{subsect:GasStorages}
In contrast to the previous subsection, in the following we consider a one factor gas market similar to the model discussed in Cummins et al.\mycite{CKM17}
\begin{equation}\label{eq:ousvg_market}
	S(t) = F(0, t) e^{h(t) + X(t)},
\end{equation}
where $X(t)$ is OU-SVG with parameters $(k, \nu, \sigma)$ and once more, because of risk-neutral arguments
\begin{equation}
		h(t) = -\frac{1}{2k\nu}\left(\dilog\left(\frac{\sigma^2\nu}{2}\right) - \dilog\left(\frac{\sigma^2\nu}{2}\,e^{-kt}\right)\right).
\label{eq:spot:rn:model2}
\end{equation} 

We then adopt this market dynamics for the pricing of a fast-churn gas storage.
To this end, denote by $C(t)$ the volume of a gas storage at time $t$
with $C_{min}\le C(t)\le C_{max}$. The holder of such an energy
asset is faced with a timing problem that consists  in
deciding when to inject, to withdraw or to do-nothing.

 Denoting $J(t,x,c)$ the value of a gas storage at time $t$
given $S(t)=x$, $C(t)=c$, one can write:
            \begin{equation}\label{eq:LSMC}
                J(t,x,c) = \sup_{u\in\mathcal{U}}\mathbb{E}\left[\int_t^T \phi_u\left(S(s) \right)ds + q\left(S(T),C(T) \right)\,\right| S(t)=x, C(t)=c\bigg],
            \end{equation}
 where $\mathcal{U}$ denotes the set of  the admissible
strategies, $u(t)\in\{-1,0,1\}$ is the regime at time $t$ such that
            \begin{equation}
            \left\{
                \begin{array}{lcll}
                     \phi_{-1}(S(t)) &=& -S(t)-K_{in} a_{in}, & \text{injection} \\
                    \phi_{0}(S(t)) &=& -K_N, & \text{do nothing} \\
                    \phi_{1}(S(t)) &=& S(t)-K_{out} a_{w} &\text{withdrawal}
                \end{array}
            \right.,
            \end{equation}
$a_{in}$ and $a_{w}$ are the injection and withdrawal rates,
$K_{in}$, $K_{out}$ and $K_N$,  respectively, represent the
costs of injection, do-nothing and withdrawal,  and $q$ takes
into account the possibility of final penalties.
Based on the Bellman recurrence equation (see Bertsekas\mycite{Bertsekas05}), one can perform the following backward recursion for $i=1,\dots,d$:
                \begin{equation}
                    J(t_i,x,c) = \sup_{k\in\{-1,0,1\}} \left\{\phi_k S(t_i) + \mathbb{E}\left[ J\left(t_{i+1},S(t_{i+1}),\tilde{c}_k\right)| S(t_i)=x, C(t_i)=c\right]   \right\} , i=1,\dots,d,
                \end{equation}
                where
                \begin{equation}
                \left\{
                \begin{array}{lll}
                    \tilde{c}_{-1} &=& \min(c+a_{in}, C_{max})\\
                    \tilde{c}_{0} &=& c\\
                    \tilde{c}_{1} &=& \min(c-a_w, C_{min}).\\
                \end{array}
                \right.
                \end{equation}
A standard approach to price gas storages is a modified
version of the Least-Squares Monte Carlo (LSMC), introduced in
Longstaff-Schwartz\mycite{LSW01}, detailed in Boogert and de
Jong\mycite{BDJ08}. With this approach, the backward recursion is obtained by defining a
finite volume grid of G steps for the admissible capacities $c$ of the plant and then apply
the LSMC methodology to the continuation value per volume step. In alternative, one may solve the recursion by adapting the method proposed by  Ben-Ameur et al.\mycite{BBKL2007} or might use the quantization method as explained in Bardou et al.\mycite{BBP07} or even FFT and Fourier techniques described for instance in Jaimungal and Surkov\mycite{JaimungalSurkov11}.

Finally, we consider  a one-year fast-churn storage with the
parameters shown in Table\myref{tab:Storage:Parameters} such that $20$
days are required to fill or empty the storage.
Our framework is similar to the one presented in Cummins et. al.\mycite{CKM17} and we select the parameters of the OU-SVG shown therein. Once again, we assume a flat forward curve that does not change the validity of our experiment because we price a fast-churn storage whose extrinsic value is dominated by the short-term decisions rather than by the seasonality of the forward curve. 
In this case, Figure\myref{fig:trajectory:gas:storage} shows one sample price-path generated with Algorithm\myref{alg:sim:svg} using the market model of Equation\refeq{eq:ousvg_market}.

The columns of Table\myref{tab:Storage:Results} have the same meaning of those presented in the analysis of the Asian option. In addition, the column \emph{CPU$^*$} reports the computational times relative to the path simulation only, using Algorithm\myref{alg:sim:svg}. Indeed, the computational cost of the standard LSMC method can be split into a path generation step and into a stochastic optimization step, with the computational cost of the latter one being independent on the price dynamics and being the dominant factor to the overall computational time (CPU in Table\myref{tab:Storage:Results}). In this example, the CPU$^*$'s are approximately, one forth of the overall time, and in particular, are almost a half of the overall time required to price an Asian option with the same number of simulations. Once again, the results show that our methodology is accurate and reasonably fast. However, although still acceptable, the CPU$^*$'s are higher that those required to simulate other generalized OU processes. For instance, one of the reasons why Cummins et al.\mycite{CKM17} investigated the use of OU-SVG was to reduce the number of model parameters compared to the mean-reverting jump-diffusion model discussed in Kjaer\mycite{Kjaer2008}. In this last case however, Cufaro Petroni and Sabino\mycite{cs20_2} have presented a fast simulation algorithm for mean-reverting jump-diffusion dynamics including the case of the mean-reverting equivalent of the Kou\mycite{Kou2002} model. These computational times are faster that those presented in this study because they almost completely cut off the time needed for the path simulation that therefore, becomes negligible compared to that of the optimization step.   

It is worthwhile noticing that the performance of standard LSMC could be improved relying on the backward simulation of the price dynamics as explained in Pellegrino and Sabino\mycite{PellegrinoSabino15} and Sabino\mycite{Sabino20}. This means that one has to design a backward simulation algorithm for the OU-VG process that will be the objective of a future research. The results of Table\myref{tab:Storage:Results} show that one should rely on at least $N_S=10000$ sample paths to get an acceptable price. Of course, this also depends on the granularity of the volume grid: we have chosen $100$ equally-spaced steps. The computational performance of the LSMC based on Algorithm\myref{alg:sim:svg} is inferior to that of other numerical techniques such as FFT based approaches. Nevertheless, it has the advantage to be applicable to any payoff function, in contrast to FFT techniques that have to be adapted to each particular contract.

\input{./Tables/StorageResults}

\section{Conclusions}\label{sec:Conclusions}
In this paper we have introduced a three-steps procedure to determine the law of the increment of generalized OU processes relying on the role of self-decomposability in the theory of such processes. Based on this machinery and the results of Qu et al.\mycite{QDZ19}, we have developed  efficient and accurate algorithms for the exact simulation of the OU-VG and
OU-SVG processes discussed in Cummins et al.\mycite{CKM17, CKM18} and Cufaro et al.\mycite{CDDI07}. The algorithms are accurate,
efficient, and have been numerically tested, with the associated performance reported in detail. In addition, our three-steps procedure simplifies some of the proofs presented in the cited papers and could be employed to find simulation algorithms of other generalized OU processes that will the object of future inquires. 

These results are instrumental to design algorithms to price derivative contracts in energy markets by MC simulation. To this end, we have considered the case of an Asian option in a market driven by the sum of a standard VG and a OU-VG process and the case of a fast-churn gas storage in a market driven by a OU-SVG process using the LSMC method of Boogert and de Jong\mycite{BDJ08}. Although, MC methods are slower than other numerical solutions, our algorithms give the possibility to compute the entire price distribution of derivative contracts. Although the parameters estimation is not the focus of our study, MC based techniques can also be a viable route to
estimate the model parameters because the likelihood method is impracticable. Moreover, they can provide a graphical evidence that calibrated parameters correspond to realistic sample paths.

It is also worth noticing that all the algorithms that we have discussed are based
on the sequential generation of processes. Therefore, a last topic deserving
further investigation is the possibility to simulate OU-VG processes backward in time
extending the results of Pellegrino and
Sabino\mycite{PellegrinoSabino15} and Sabino\mycite{Sabino20}.
    
		\bibliographystyle{plain}
    \bibliography{Sabino_OU_VG}
\end{document}

%% file: Tables/OU_VG.tex
\begin{table}[ht!]
    \centering\scriptsize
		\resizebox{\textwidth}{!}{
        \begin{tabular}{*{10}{|cc|cc|cc|cc|cc}}
					\hline
					\multicolumn{2}{|c}{} &  \multicolumn{2}{|c}{$\EXP{X(T)} = 0.0490$} & \multicolumn{2}{|c}{$\VAR{X(T)}=0.0185$} & \multicolumn{2}{|c}{$\SK{X(T)}=0.529$} & \multicolumn{2}{|c|}{$\KUR{X(T)}=4.689$}\\			
					\hline	
					\multicolumn{10}{|c|}{$T=1/5, \Delta t =1/5$}\\					
					\hline
					$N_S$ & CPU & MC & error \% & MC & error \% & MC & error \% & MC & error \%	\\							
					\hline
					$2500$ & $0.05$ & $0.0475$ & $3.14$ & $0.0188$ & $1.68$ & $0.537$ & $1.52$ & $4.786$ & $2.05$\\
					$10000$ & $0.20$ & $0.0491$ & $0.17$ & $0.0182$ & $1.57$ & $0.502$ & $5.08$ & $4.637$ & $1.12$\\
					$40000$ & $0.77$ & $0.0489$ & $0.31$ & $0.0187$ & $1.07$ & $0.529$ & $0.07$ & $4.703$ & $0.30$\\
					$160000$ & $3.06$ & $0.0494$ & $0.82$ & $0.0186$ & $0.48$ & $0.535$ & $1.17$ & $4.666$ & $0.49$\\
					$640000$ & $12.32$ & $0.0489$ & $0.19$ & $0.0185$ & $0.02$ & $0.522$ & $1.39$ & $4.646$ & $0.93$\\
					$2560000$ & $49.39$ & $0.0491$ & $0.18$ & $0.0185$ & $0.11$ & $0.527$ & $0.38$ & $4.678$ & $0.24$\\
					\hline
					\multicolumn{2}{|c}{} &  \multicolumn{2}{|c}{$\EXP{X(T)}=0.2266$} & \multicolumn{2}{|c}{$\VAR{X(T)}=0.0793$} & \multicolumn{2}{|c}{$\SK{X(T)}=0.238$} & \multicolumn{2}{|c|}{$\KUR{X(T)}=3.342$}\\			
					\hline
					\multicolumn{10}{|c|}{$T=1, \Delta t =1/5$}\\
					\hline
					$2500$ & $0.28$ & $0.2236$ & $1.31$ & $0.0786$ & $0.97$ & $0.214$ & $9.81$ & $3.220$ & $3.67$\\
					$10000$ & $1.02$ & $0.2270$ & $0.18$ & $0.0800$ & $0.83$ & $0.248$ & $4.31$ & $3.332$ & $0.30$\\
					$40000$ & $4.02$ & $0.2272$ & $0.28$ & $0.0802$ & $1.15$ & $0.236$ & $0.70$ & $3.315$ & $0.83$\\
					$160000$ & $16.03$ & $0.2268$ & $0.09$ & $0.0794$ & $0.10$ & $0.230$ & $3.42$ & $3.316$ & $0.78$\\
					$640000$ & $64.46$ & $0.2261$ & $0.20$ & $0.0793$ & $0.02$ & $0.239$ & $0.55$ & $3.341$ & $0.03$\\
					$2560000$ & $258.51$ & $0.2265$ & $0.06$ & $0.0793$ & $0.10$ & $0.238$ & $0.30$ & $3.338$ & $0.12$\\
					\hline
        \end{tabular}				
		}
    \scriptsize
    \caption{\footnotesize{CPU times in seconds and comparison  among the true $\EXP{X(T)}$, $\VAR{X(T)}$, $\SK{X(T)}$ and $\KUR{X(T)}$ of a OU-VG process with $(k, \theta, \nu, \sigma, X(0)) =(0.2, 0.25, 0.1, 0.3, 0)$ and their relative estimated values with $N_S$ MC scenarios.}}\label{tab:ou:vg:MC}
\end{table}

%% file: Tables/OU_Sym_VG.tex
\begin{table}[ht!]
    \centering\scriptsize
		\resizebox{\textwidth}{!}{
        \begin{tabular}{{|c|c|cc|cc||c|cc|cc|}}
					\hline
					& \multicolumn{5}{c||}{$T=1/5, \Delta t=1/5$} & \multicolumn{5}{c|}{$T=1, \Delta t=1/5$}\\					
					\hline	
					\multicolumn{2}{|c}{} &  \multicolumn{2}{|c|}{$\VAR{X(T)}=0.0077$} & \multicolumn{2}{c||}{$\KUR{X(T)}=6.84$} & &
					\multicolumn{2}{c}{$\VAR{X(T)}=0.0328$}  & \multicolumn{2}{|c|}{$\KUR{X(T)}=3.780$}\\			

					\hline
					$N_S$ & CPU & MC & error \% & MC & error \% & CPU & MC & error \% & MC & error \%	\\							
					\hline
$2500$ & $0.04$ & $0.0079$ & $1.96$ & $6.89$ & $0.67$ & $0.20$ & $0.0303$ & $7.69$ & $3.341$ & $11.61$\\
$10000$ & $0.12$ & $0.0075$ & $2.73$ & $6.66$ & $2.62$ & $0.64$ & $0.0328$ & $0.07$ & $3.748$ & $0.84$\\
$40000$ & $0.53$ & $0.0078$ & $0.66$ & $6.96$ & $1.79$ & $2.50$ & $0.0327$ & $0.29$ & $3.853$ & $1.94$\\
$160000$ & $2.04$ & $0.0077$ & $0.73$ & $6.91$ & $1.03$ & $10.16$ & $0.0326$ & $0.71$ & $3.717$ & $1.66$\\
$640000$ & $8.06$ & $0.0077$ & $0.18$ & $6.78$ & $0.88$ & $39.98$ & $0.0327$ & $0.38$ & $3.793$ & $0.35$\\
$2560000$ & $32.10$ & $0.0078$ & $0.32$ & $6.85$ & $0.16$ & $159.84$ & $0.0328$ & $0.04$ & $3.778$ & $0.05$\\

					\hline
        \end{tabular}				
		}
    \scriptsize
    \caption{\footnotesize{CPU times in seconds and comparison among the true $\VAR{X(T)}$ and $\KUR{X(T)}$ of a OU-Symmetric VG process with $(k, \nu, \sigma, X(0)) =(0.2162, 0.256, 0.201, 0)$ and their relative estimated values with $N_S$ MC scenarios}}\label{tab:ou:sym:vg:MC}
\end{table}

%% file: Tables/Asian_Results.tex
\begin{table}[!htb]
\caption{Asian option in a $2$-Factors market dynamics}\label{tab:Asian2Factor}
\begin{subtable}{.3\linewidth}
		\centering
			\begin{tabular}{|c|c|}
				\hline
				Parameter & Value \\
				\hline
				$F(0,0)$ & $15$\\
				$K$ & $15$\\
				$T$ & $1$\\
				$d$ & $360$\\
				$\kappa_1$ & $0.1859$\\
				$\theta_1$ & $0.05$\\
				$\nu_1$ & $0.4513$\\
				$\sigma_1$ & $0.203$\\
				$\theta_2$ & $0.1$\\
				$\nu_2$ & $0.2$\\
				$\sigma_2$ & $0.3$\\
				\hline
		\end{tabular}
\caption{Parameters}\label{tab:Asian:Parameters}
\end{subtable}%
\begin{subtable}{.70\linewidth}
			\centering
			\begin{tabular}{|c|c|c|c|c|c|}
				\hline
				$N_S$ & CPU & price & stdev & error & \%-error\\
				\hline
				$1000$ & $7.59$ & $1.219$ & $2.24$ & $0.071$ & $5.82$\%\\
				$10000$ & $74.87$ & $1.229$ & $2.34$ & $0.023$ & $1.90$\%\\
				$20000$ & $149.03$ & $1.227$ & $2.26$ & $0.016$ & $1.30$\%\\
				$50000$ & $386.10$ & $1.231$ & $2.26$ & $0.010$ & $0.82$\%\\
				$100000$ & $743.03$ & $1.236$ & $2.27$ & $0.0072$ & $0.58$\%\\
				\hline
		\end{tabular}
		\caption{Results}\label{tab:Asian:Results}
\end{subtable}
\end{table}

%% file: Tables/StorageResults.tex
\begin{table}[!htb]
	\caption{Gas Storage in a $1$-Factors market dynamics}\label{tab:StorageFactor}
	\begin{subtable}{.3\linewidth}
		\centering\
			\begin{tabular}{|c|c|}
				\hline
					Parameter & Value \\
				\hline
				$F(0,0)$ & $15$\\
				$T$ & $1$\\
				$d$ & $360$\\
				$\kappa$ & $0.2162$\\
				$\nu$ & $0.2560$\\
				$\sigma$ & $0.2021$\\				  
				 $C(0) $   & $0$ \\
				 $C(T) $ & $0$\\
				 $a_{in}$ & $1$\\
				 $a_{w}$ & $1$\\
				 $C_{max}$ & $20$\\
				\hline
			\end{tabular}
			\caption{Parameters}\label{tab:Storage:Parameters}
	\end{subtable}%
	\begin{subtable}{.70\linewidth}
			\centering
			\begin{tabular}{|c|c|c|c|c|c|c|}
				\hline
				$N_S$ & CPU & CPU$^*$ & price & stdev & error & \%-error\\
				\hline
				$1000$ & $12.92$ & $4.74$ & $8.79$ & $0.33$ & $0.011$ & $0.12$\%\\
				$10000$  & $166.84$ & $47.08$ & $5.20$ & $0.49$ & $0.005$ & $0.09$\%\\
				$20000$  & $307.39$ & $92.16$ & $5.24$ & $0.34$ & $0.002$ & $0.05$\%\\
				$50000$  & $823.98$ & $139.01$ & $5.10$ & $0.37$ & $0.002$ & $0.03$\%\\
				$100000$  & $1557.09$ & $465.41$ & $5.13$ & $0.36$ & $0.001$ & $0.02$\%\\
				\hline		
			\end{tabular}
    \caption{Storage Results}\label{tab:Storage:Results}
	\end{subtable}
\end{table}